\documentclass[review]{elsarticle}
\usepackage{amsmath}
\usepackage{amsthm}
\usepackage{amssymb}

\usepackage{subfig}
\usepackage{graphicx}
\usepackage{booktabs}
\usepackage{braket}
\usepackage{enumerate}
\usepackage[title]{appendix}
\usepackage[colorlinks=true, allcolors= black]{hyperref}
\usepackage{pgfplots}
\newtheorem{lemma}{Lemma}
\newtheorem{theorem}[lemma]{Theorem}
\newtheorem{remark}{Remark}
\newtheorem{definition}{Definition}
\usepackage[ruled]{algorithm2e}
\usepackage[most]{tcolorbox}
\newtcbtheorem{mytcbprob}{Problem}{}{problem}

\title{Asymptotically Optimal Synthesis of Reversible Circuits}
\begin{document}

\begin{frontmatter}
\author{Xian Wu\corref{email}}
\author{Lvzhou Li\corref{mycorrespondingauthor}}
\address{Institute of Quantum Computing and Software, School of Computer and Engineering, Sun Yat-sen University, Guangzhou 510006, China}
\cortext[email]{wuxian3@mail2.sysu.edu.cn}
\cortext[mycorrespondingauthor]{Corresponding author. lilvzh@mail.sysu.edu.cn}

\begin{abstract}

Reversible circuits have been  studied extensively and intensively, and  have plenty of  applications in various areas, such as digital signal processing, cryptography, and especially
quantum computing.  In 2003, the lower bound $\Omega(2^n n/\log n)$ for the synthesis of  $n$-wire reversible circuits was proved. Whether this lower bound {\color{black} has a matching upper bound}  was   listed as one of the future challenging open problems in the  survey (M. Saeedi and I. L Markov, ACM Computing Surveys, 45(2):1–34, 2013).
In this paper we  propose an algorithm to implement an arbitrary $n$-wire reversible circuit with no more than $O(2^n n/\log n)$ elementary gates, and thus  close the open problem.

\end{abstract}

\begin{keyword}
Quantum Computing, Reversible Circuit Synthesis, Circuit Optimization.
\end{keyword}

\end{frontmatter}

\section{Introduction}
\subsection{Background}
Reversible circuits play a crucial role in both classical   and  quantum computing, and one can refer to   the excellent survey \cite{saeedi2013synthesis} for full information. In classical computing, many cryptographic algorithms \cite{shi2000bit, hilewitz2008fast, thorat2018implementation} use bit permutation instructions, whose performance can be improved by using better reversible computation synthesis methods. In addition, signal processing often requires reversible transforms, and the reversible data hiding  algorithm is useful in  signal processing \cite{gao2015reversible, xiong2018integer}.
In quantum  computing,   reversible circuits are  indispensable. First,  some prominent quantum algorithms  use  blocks of reversible circuits \cite{beckman1996efficient, markov2012constant, niemann2015synthesis}, for example, the modular exponentiation in Shor's algorithm \cite{shor1994algorithms}. Second, reversible circuit synthesis has a great significance in quantum computing, as pointed out in \cite{abdessaied2016reversible} that the non-optimality of quantum circuits derives mainly from the exponential complexity of reversible circuits. Third, reversible circuits are a fundamental part of quantum compilation, as mentioned in  \cite{zulehner2020introducing} that "Design of Boolean Components" is one of two tasks to compile quantum algorithms.

Reversible logic synthesis  is the process of generating a compact reversible circuit from a given specification. More formally, given a reversible Boolean function or a permutation matrix, reversible logic synthesis considers how to implement it by a reversible circuit.  A frequently used cost is the number of elementary gates (also known as size), which will also be used in this paper. 
Furthermore, when incorporating redundant input-output line-pairs (referred to as ancillary bits or temporary storage channels) in the synthesis process, the number of elementary gates needed may decrease; however, this unavoidably leads to additional space overhead. The utilization of ancillary bits is also a crucial factor to consider in circuit synthesis, particularly in the realm of quantum compilation. Therefore, it is advisable to minimize their usage.

By using Lupanov representation, one can build a circuit for any boolean function $f:\{0,1\}^n \rightarrow \{0,1\}$ with no more than $O(2^n/n)$ elementary gates. 
{\color{ black} It thus implies that for an arbitrary $n$-bit reversible function, one can construct a circuit with no more than $O(2^n)$ elementary gates. Furthermore, as shown in \cite{nielsen2002quantum}, every irreversible elementary gate such as  the AND gate can be simulated by  a Toffoli gate with ancillary bits. Thus, for an arbitrary $n$-bit reversible function, one  can construct a reversible circuit with no more than $O(2^n)$ gates. However,
this trivial idea has the defect that exponential ancillary bits are required.
In this paper, we will propose a reversible logic
synthesis  approach which dose not  employ any ancillary bit. }

\subsection{Related work}

In 2003, Ref. \cite{shende2003synthesis} proved  that an $n$-wire reversible circuit needs at least $\Omega(2^n n/\log n )$ gates in the worst case, no matter whether classical or quantum elementary gates are used (see Lemma \ref{lowerbound_lemma} for details). Whether this lower bound {\color{ black} has a matching upper bound} was   listed as one of the future challenging open problems in the  survey \cite{saeedi2013synthesis}.
There are some works devoted to approaching the lower bound  $\Omega(2^n n/\log n )$. Ref. \cite{abdessaied2014upper} used Young subgroups and exclusive sum-of-products (ESOPs) to implement reversible circuits, but the upper bound is  $O(n2^n)$.  The best upper bound   $8.5n 2^n +o(2^n)$ was obtained by a cycle-based method \cite{saeedi2010reversible}, which is better than the previous bounds in terms of coefficients, but the order of magnitude remains the same. Note that these results above are without the use of auxiliary bits and require the use of quantum gates. With the help of auxiliary bits,   Binary Decision Diagram (BDD) -based synthesis can construct reversible circuits much faster, yet they still require $O(2^n n)$ elementary gates in the worst case \cite{wille2009bdd}.


\subsection{Motivation}
Despite the effort of the above works, the problem whether the lower bound  $\Omega(2^n n/\log n )$ {\color{ black} has a matching upper bound} has been left open for about 20 years. The significance of this problem can be seen as follows. (i) In theory, the pursuit of tight complexity bounds has been an eternal topic in theoretical computer science, and a lot of works  are to continuously improve some complexity bounds little by little. (ii) In experiment, optimizing the size  of a quantum circuit  is  critical  for the success of  experiments, and even a constant multiple improvement of the circuit size  is very valuable. Before the realization of large-scale error correction, the noise of quantum devices is unavoidable, and the noise impact of the entire quantum circuit will increase monotonically with the circuit size. Therefore, it is of great importance to optimize the  size of quantum circuits.


\subsection{Contributions}
In this paper, we will present a new synthesis method that can achieve the lower bound $\Omega(2^n n/\log n)$, thus closing an open problem.
\begin{theorem}\label{maintheorem}
  {\color{ black} Any $n$-bit reversible function can be computed by a reversible  circuit with no more than $O(2^n n/\log n)$ elementary gates and  without the use of ancillary bits.}
    \label{main-theorem}
\end{theorem}

\begin{remark}\label{gate_remark}
 The elementary gates we use here are $\{NOT, S, T, H, CNOT, controlled-R_x(\theta), controlled-V_i, controlled-V_i^\dagger\}$, where $\theta \in \{\pi, \pm\frac{\pi}{2}, \pm\frac{\pi}{4}, \cdots, \pm\frac{\pi}{2^{n-3}}\}$ and  $V_i$ satisfies $(V_i)^{2^i} = NOT$ for $i \in \{1, 2, \cdots, n-2\}$.  This gate library meets the requirement in Lemma \ref{lowerbound_lemma}.  One can further decompose these controlled gates into several CNOT gates and single-bit gates, but for convenience, we treat them here as elementary gates. Actually, we can construct any $n$-wire reversible circuit with the NCT gate library \footnote{ The NCT gate library means the gate set of NOT, CNOT, and Toffoli.} and no more than one $C^{n-1}NOT$ gate \cite{shende2003synthesis} (see Definition \ref{CNOT}). NCT gate can be constructed with only $\{S, T, H, CNOT\}$, but the  $C^{n-1}NOT$ gate requires extra $\{controlled-R_x(\theta), controlled-V_i, controlled-V_i^\dagger\}$, adding up to $O(n)$ extra elementary gates \cite{da2022linear}.


\end{remark}

The main idea behind the proof of Theorem \ref{maintheorem} is to recognize that if the number of inputs $x$ such that $f(x) \neq x$ is small enough (roughly $2^{n/2}$), then it is easy to synthesize $f$ (essentially using swaps). Therefore, the crucial step is to discover a new permutation, denoted as $g$, that, when composed with the original permutation, results in a small number of $x$ instances where $g(f(x)) \neq x$. Furthermore, we will demonstrate that such a permutation $g$ requires no more than $O(2^n n/\log n)$ elementary gates. This achievement is made possible by employing specific structural lemmas (refer to Lemma \ref{lemma-matrix}, \ref{lemmaUV}, and \ref{id-lemma}), which serve as foundational elements for our argument.

\subsection{Organization}
The remainder of this document is organized as follows. In Section \ref{sec2}, we provide essential background information and formal definitions. Section \ref{sec:opt} presents our synthesis methodology. In Section \ref{sec4}, we introduce an algorithm that is an enhancement of the approach described in Section \ref{sec:opt}, yielding improved performance compared to existing algorithms. Finally, we conclude in Section \ref{secCon}.

\section{Preliminaries}\label{sec2}
In this section, we shall introduce some notations and definitions related to reversible circuits and quantum gates. 

Let $ \mathbb{F}_2$ denote the two-element field.
For an $n$-bit  non-negative integer $i$, we shall employ $i$ as its binary representation.  $\Vec{i}$  shall depict its n-bit binary string in a column format, while $\ket{i}$  represents the $2^n$-dimensional vector with only the $(i+1)$-th element being 1 and the others being 0. To be more precise, let $i = \sum_{j=1}^{n} a_j 2^{n-j}$. Then we have $\Vec{i}=\begin{bmatrix}
    a_1&
    a_2&
    \cdots&
    a_n
\end{bmatrix}^T$ and  $\ket{i}=\ket{a_1}\otimes \ket{a_{2}} \cdots\otimes \ket{a_n}$.

The $n\times n$ identity matrix can be expressed as $I=\begin{bmatrix}
    \Vec{e}_{1}&
    \Vec{e}_{2}&
    \cdots&
    \Vec{e}_{n}
\end{bmatrix}$, where $\Vec{e}_{t}$ is a vector whose $t$-th element is 1, and the others are 0. These basis vectors play an important role in our method.

{\color{ black}
A function $f:\{0,1\}^n \rightarrow \{0,1\}^n$ is said to be reversible if it is a bijection. An $n$-bit reversible function can be represented by a $2^n\times 2^n$ permutation matrix $P$ such that  if  $f(i)=x_i$ for $ i \in \{0,1\}^n$, then we have $P\ket{i} = \ket{x_i}$.}



{\color{ black}
\begin{definition}\label{def1}
    For a given permutation matrix P, let $S_P=\{i : P\ket{i} \neq \ket{i}\}$ be the set of integers that is changed under  $P$, and $|S_P|$ be the number of elements in $S_P$.
\end{definition}

Obviously, for any $i$, if $i\in S_P$, then we have $x_i \in S_P$. 
Usually, $|S_P|$ is positively correlated with the cost of the circuit implementing $P$. 

\begin{lemma}\label{no0lemma}
    For a $2^n \times 2^n$ permutation matrix P, if $P\ket{0} \neq \ket{0}$, then we can use no more than n NOT gates to implement an operator $T$, such that $TP\ket{0} = \ket{0}$.
\end{lemma}

\begin{proof}
   If the $i$-th bit of $P\ket{0}$ is 1, then we apply a NOT gate on the $i$-th bit. Since it only has $n$ bits, it costs no more than $n$ NOT gates.
\end{proof}

This lemma makes sure that we don't need to consider $\ket{0}$, which is important in Section~\ref{sec:opt}.


In this work, we consider how to implement a reversible function by a reversible circuit which is made of reversible gates. A reversible circuit has the same number of
input and output wires, and it is said to be $k$-bit or $k$-wire, if it has $k$ input and output wires.

}

Linear reversible  circuits form an important sub-class of reversible circuits, and here we will give their formal definition.
\begin{definition}
    A reversible function $f:\{0,1\}^n \rightarrow \{0,1\}^n$ is called linear   if $\forall x,y \in \{0, 1 \}^n $, $f(x\oplus y)=f(x)\oplus f(y)$, where $\oplus$ denotes bitwise XOR.  A  reversible circuit is called linear, if the function it realizes is linear.
\end{definition}
 
For an $n$-bit linear reversible function $f$ (or an $n$-bit linear reversible circuit), we can use an $n \times n$ reversible matrix $R$ over $ \mathbb{F}_2$ to represent it such that $R\Vec{i} =\Vec{f(i)} = \Vec{x}_{i}$. 
Since all linear reversible circuits can be constructed with CNOT gates alone and adding a small number of additional single and double qubit gates does not result in an order of magnitude improvement (as can be seen in the proof of the lower bound of reversible circuits), we typically only use CNOT gates to construct linear reversible circuits.
From Ref. \cite{patel2008optimal}, we know the following result.
\begin{lemma}\label{thorem:res}
Any $n$-bit linear reversible circuit (or $n \times n$ reversible matrix) can be constructed by using only CNOT gates and no more than $O(n^2/\log n)$ CNOT gates are required.
\end{lemma}

One straightforward approach is to leverage the outcomes of linear reversible circuits to optimize reversible circuits. However, the actual implementation is not as intuitive. We need to introduce certain techniques. In Section~\ref{sec:opt}, we shall show how to use linear reversible circuits to help construct reversible circuits.

\begin{definition}\label{CNOT}
    A C$^m$NOT gate denotes a NOT gate with $m$ control bits:  if the $m$ control bits are all 1, then the target bit flips. Otherwise, leaves the target bit unchanged.
\end{definition}
For $m =0, 1, 2$, the gates are NOT, CNOT, Toffoli, respectively. The three gates are the most common reversible gates and they compose the  NCT gate library, from  which one can construct any $n$-bit reversible circuit with at most one 
 extra C$^{n-1}$NOT gate \cite{desoete2002reversible}.

\begin{lemma}\label{costofCNOT}
    In an $n$-bit circuit, the cost of C$^m$NOT gate is linear with $m$ except for $m=n-1$. The cost of C$^{n-1}$NOT gate is $O(n^2)$. 
\end{lemma}
\begin{proof}
    For $n\geq 5$ and $m \in \{3,4,\cdots, n-2 \}$, a C$^m$NOT gate can be implemented by using a linear cost of Toffoli gates \cite{barenco1995elementary}. And each Toffoli gate can be constructed in a constant cost with gate library $\{NOT, S, T, H, CNOT\}$.

    For C$^{n-1}$NOT gate, the best cost so far is $O(n^2)$ elementary gates \cite{da2022linear}.
\end{proof}

If we desire a C$^m$NOT gate that effectively permits a control bit to assume the value 0, we simply incorporate a single NOT gate prior to and another following the C$^m$NOT gate. These are known as general C$^m$NOT gates (or MCT gate), and the added expense amounts to no more than $2m$ elementary gates. 

The fundamental concept underlying our synthesis methodology is to disassemble a reversible circuit into these versatile C$^m$NOT gates. By quantifying the known expense of these gates, we can assess the overall cost of the circuit.

\begin{definition}
  When we use $q_i$($\Bar{q_i}$) to denote the control bit,  the controlled gate works only when the $i$-th bit is $1 (0)$.
\end{definition}

\begin{lemma}\label{lemma2.2}
   Given numbers $a, b\in\{0, 1\}^n$ with $a \neq b$, we can use no more than $O(n^2)$ elementary gates to construct an operator $T$ such that
    \begin{equation}
        T\ket{x}=\left\{
        \begin{array}{rcl}
            \ket{b}, & & x = a\\
            \ket{a}, & & x = b\\
            \ket{x}, & & else
         \end{array} \right.
    \end{equation}
\end{lemma}

\begin{proof}
    Let $\Vec{c} = \Vec{a} \oplus \Vec{b}$, and $c_i$ be the $i$-th element of $\Vec{c}$. Since $a \neq b$, there exists $t \in \{1, 2, \cdots, n\}$ such that $c_t = 1$.

    For each $s\in \{1, 2, \cdots, n\}$ such that $c_s = 1$ and $s \neq t$, apply a CNOT gate on $q_s$, whose control bit is $q_t$. It's obvious that it will not cost more than $n-1$ CNOT gates. After these gates, $a$ and $b$ only differ in $q_t$. 
    {\color{ black}$q_t$ should be applied with a general C$^ {n-1}$NOT gate, using the remaining bits as control bits.}
    Then apply the CNOT gates above in the reverse order, obtaining the operator $T$.

    The total cost is $O(n)+O(n^2)=O(n^2)$.
\end{proof}

For each $i\in S_P$, we can construct a $T$ which swaps $i$, $x_i$ and leaves the others unchanged. Then we have $TP\ket{i} = \ket{i}$ and $TP\ket{j} = \ket{j}$  for any $j$ satisfying $P\ket{j}=\ket{j}$. If we can reduce $|S_P|$ to small enough, then the gate complexity of implementing $T$ can be neglected in the synthesis process.

\begin{definition}
 An $n$-wire circuit is said to be one-gate, if it consists of only one elementary gate.
\end{definition}

For the convenience of readers, we recall the lower bound of reversible logic synthesis proved in \cite{shende2003synthesis}.

\begin{lemma}
    \label{lowerbound_lemma}
   Given an elementary gate library $L$, if the number of one-gate circuits on $n$ wires generated by $L$ is polynomial with $n$,  then in the worst case  an $n$-bit  reversible circuits requires $ \Omega(2^n n/\log n)$ gates from $L$.
\end{lemma}

\begin{proof}
    It's a special case of Lemma 8 in \cite{shende2003synthesis}.

  {\color{ black}
  Let $b$ be the number of one-gate circuits on $n$ wires generated by $L$. 
  Let $d$ be the gate number of the largest $n$-bit  reversible circuit generated by $L$.
  At most $b^d $  matrices can be generated by such a circuit.
Let $M$ be the number of  all the $n$-bit reversible circuits with  $d$  gates. Then we have
    \begin{equation}
        b^d\geq M \geq (2^n)!  \ ,
    \end{equation}
    where $(2^n)!$ is the number of all the $n$-bit reversible functions. Therefore,
    \begin{equation}
        d\geq \sum_{i=1}^{2^n}{\log_b i} = \frac{1}{\ln b}\sum_{i=1}^{2^n}{\ln i}  \ .
    \end{equation}
    Since $b=poly(n)$, we have $d = \Omega(2^n n/\log n)$. 
    }
\end{proof}

\begin{remark}
Note that any  gate library $L$ with  $poly(n)$   one-qubit gates and two-qubit gates   satisfies the requirement in  Lemma \ref{lowerbound_lemma}. In addition, the lower bound $ \Omega(2^n n/\log n)$ still holds when  $poly(n)$ ancillary qubits are allowed. 
\end{remark}

\section{Asymptotically Optimal Synthesis}\label{sec:opt}

In this section,  we are going to prove Theorem \ref{main-theorem}. More specifically, we will prove that for an arbitrarily given $2^n\times 2^n$ permutation matrix $P$ or  $n$-bit reversible function $f$, we can implement it by a reversible circuit with no more than $O(2^n n/\log n)$ elementary gates.


In the following, we should remember that $P$ is a $2^n\times 2^n$ permutation matrix, $f$ means the reversible function corresponding to $P$ and $S_P=\{i : P\ket{i}\neq \ket{i}\}$. Also note that all matrices and vectors are over $ \mathbb{F}_2$.

We will construct a series of $Q_1, Q_2,\cdots, Q_d$ such  that  $|S_{ Q_d\cdots Q_2 Q_1P}| \leq 2^{n/2-1}$, where $d$ is no more than $O(2^n/n)$ and each $Q_i$ costs no more than $O(n^2/\log n)$ gates. This result is formally stated in Theorem \ref{mainswap}, which leads to the main result---Theorem \ref{maintheorem}.  Before proving it, we need several technical lemmas  showing how to construct  $Q_i$.
For the reader's convenience, the relationship among the lemmas and theorems obtained in this paper is shown in Figure \ref{p_lemma}.

\begin{figure}[h]
    \centering
    \includegraphics[width=0.8\textwidth]{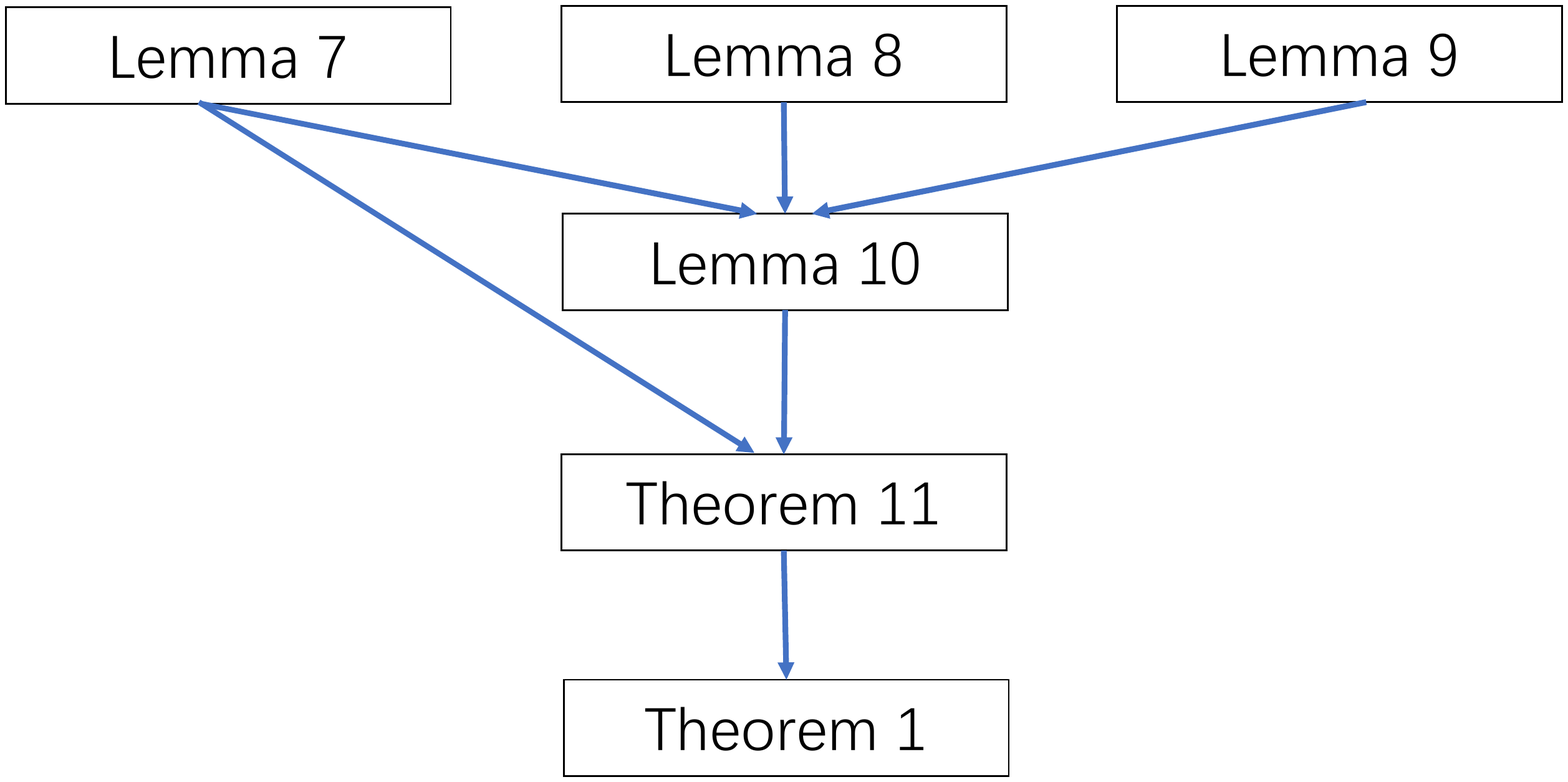}
    \caption{ Relationship among the lemmas and theorems obtained in this paper. }
    \label{p_lemma}
\end{figure}

\begin{lemma}\label{lemma-matrix}
Let $m=\lfloor\log_2n\rfloor-1$. When $|S_P|>2^{n/2-1}$,  $2^{m-1}$ pairs of $(i, x_i)$ from $S_P$ can be selected to construct $V=
\begin{bmatrix}
    \Vec{i}_{1}&
    \Vec{x}_{i_1}&
    \Vec{i}_{2}&
    \Vec{x}_{i_2}&
    \cdots&
    \Vec{i}_{2^{m-1}}&
    \Vec{x}_{i_{2^{m-1}}}
\end{bmatrix}
$ and
$
U=
\begin{bmatrix}
    \Vec{i}_{1}&
    \Vec{z}_{1}&
    \Vec{i}_{2}&
    \Vec{z}_{2}&
    \cdots&
    \Vec{i}_{2^{m-1}}&
    \Vec{z}_{2^{m-1}}
\end{bmatrix}$, such that the columns of $U$ are linearly independent, with $\Vec{z}_{1}=\Vec{x}_{i_1}$, and for $ t \in \{1, 2,\cdots, 2^{m-1}\}$,  neither $v_{2t-1}$ nor $v_{2t}$ is a linear combination of $\{u_1, u_2, \cdots, u_{2t-2}\}$, where $u_t(v_t)$ denotes the $t$-th column of $U(V)$.
\end{lemma}

\begin{proof}
At the beginning, let $U = \begin{bmatrix}& \end{bmatrix}$, $V = \begin{bmatrix}& \end{bmatrix}$. Then we add vectors into the two matrices from left to right by repeating the following iteration $2^{m-1}$ times. 
    
In the $t$-th iteration ($t=1,2,\cdots, 2^{m-1}$), we take the following three steps. 
\begin{enumerate}[(1)]
    \item 
    Select a pair $(i_t, x_{i_t})$ from $S_P$  such that  both $\Vec{i}_{t}$ and $\Vec{x}_{i_t}$  are  not a linear combination of the columns of $U$ (it  will be proved  later such a pair can always be found if $|S_P|>2^{n/2-1}$).

    \item 
    Append $\Vec{i}_{t}$ and $\Vec{x}_{i_t}$  to $V$.

    \item 
    Append $\Vec{i}_{t}$ to $U$ and check whether  the columns of  the obtained $U$ together with  $\Vec{x}_{i_t}$ are linearly independent. If yes, then append $\Vec{x}_{i_t}$ to $U$. Otherwise, select a vector $l$   such that the the columns of $U$ together with $l$ are linearly independent and then append it to $U$.
    
\end{enumerate}

Now we explain why the pair $(i_t, x_{i_t})$ in the first step above can always be found. 
When we are selecting the $t$-th pair $(i_t, x_{i_t})$, it means  $t-1$ pairs have been selected and  there are $2(t-1)$ columns in  $U$. Thus, the number of all the linear combinations of the columns of $U$ is $2^{2t-2}$, and each  vector $\Vec{u}$ from this set will lead to that $(u, f(u))$, $(f^{-1}(u),u)$ can not be selected. Thus, there are at most $2\times2^{2t-2}=2^{2t-1}$ pairs that can not be the candidate for the $t$-th pair.  Therefore, when $|S_P|>2^{2t-1}$, the $t$-th pair $(i_t, x_{i_t})$ can be selected. Furthermore,  when $|S_P|>2^{n/2-1}\geq 2^{2^m-1}$, $2^{m-1}$ pairs of $(i, x_i)$ can always be found.
    
\end{proof}
Based on the construction of $V$, it has the following properties:
\begin{enumerate}[(1)]
\item $\forall t \in \{1,2,3,5,7,\cdots, 2^{m}-1\}$, $v_t=u_t$.

\item $\forall t \in \{4,6,8,\cdots, 2^m\}$, 
    either $v_t=u_t$ or $v_t =u_{t-1} \oplus \sum_{j=1}^{t-2}{k_j u_j}$,  {\color{ black}where at least one $k_j$ is equal to 1, with each $k_j \in \{0, 1\}$}.

\end{enumerate}
Here we explain item (2). 
{\color{ black}Lemma \ref{lemma-matrix} ensures that $\{u_1, u_2, \cdots, u_{t-2}, v_{t-1}\}$ and $\{u_1, u_2, \cdots, u_{t-2}, v_t\}$ are both linearly independent sets. If $v_t\neq u_t$, then $\{u_1, u_2, \cdots, u_{t-2}, u_{t-1},v_t\}$ are not linearly independent set.
Then we have the fact that $v_t \oplus u_{t-1}=v_t \oplus v_{t-1}=\sum_{j=1}^{t-2}{k_j u_j}$ and at least one $k_j$ equals 1. }
Thus, by adding $u_{t-1}$ to the two sides of the equation, we get item (2).

Furthermore, it's easy to augment $U$ to an $n \times n$ linearly independent matrix over $ \mathbb{F}_2$ by adding some columns. Correspondingly,   $V$  is also augmented by 
adding the same columns. The augmented matrices are still denoted by $U$ and $V$.

\begin{lemma}\label{lemmaUV}
    Once we get $U, V$ above, we can construct an $n \times n $ reversible matrix $R_1$  over $ \mathbb{F}_2$, such that $R_1 U=I$. 
    As a result,  $V':=R_1 V$ is an upper triangular matrix having  the following properties:
    \begin{enumerate}[(1)]
    \item  $\forall t \in \{1,2,3,5,7,\cdots, 2^{m}-1\}$, $v'_t=\Vec{e}_{t} $, where $v'_t$ is the $t$-th column vector of $V'$.

     \item   $\forall t \in \{4,6,8,\cdots, 2^m\}$, if $v_t = u_t$, then $v'_t=\Vec{e}_t$; otherwise, $v'_t=\Vec{e}_{t-1} \oplus \sum_{j=1}^{t-2}{k_j \Vec{e}_{j}}$,  {\color{ black} where $k_j$ belongs to $\{0, 1\}$ and at least one $k_j$ equals 1}.
    \end{enumerate}
   
\end{lemma}
 
\begin{proof}
   First it is easy to  get item (1) from the fact that $R_1 u_t=\Vec{e}_{t}$.
    
    Because of the property (2) of $V$, $\forall t \in \{4,6,8\cdots, 2^m\}$, if $v_t\neq u_t$, then $R_1 v_t = R_1 u_{t-1} \oplus \sum_{j=1}^{t-2}{k_j R_1 u_j}=\Vec{e}_{t-1} \oplus \sum_{j=1}^{t-2}{k_j \Vec{e}_{j}}$. Therefore, item  (2) holds. It's easy to check that the $k_j$ of $v'_t$ in property (2) is the $j$-th element of $v'_t$.
    {\color{ black}For $t \in \{2^{m} + 1, 2^m +2,\cdots, n\}$, we have $u_t=v_t$ and $R_1v_t = \Vec{e_t}$.}

    Based on these properties, we have $v'_{t',t}=0$ for any $t'>t$, where $v'_{i,j} $ means the $i$-th row and $j$-column element of $V'$. 
    Thus $V'$ is an upper triangular matrix. 
    
\end{proof}

\begin{lemma}\label{id-lemma}
    Once we get $V'$ from Lemma \ref{lemmaUV}, we can use no more than $O(n^2/\log n)$ elementary gates to transform $V'$ to the identity matrix $I$.
\end{lemma}

\begin{proof} We first show that one can use no more than $O(n)$  elementary gates to transform $V'$  to an upper triangular matrix with the main diagonal elements being $1$. More specifically, we address  each column $v'_t$ of $V'$ from left to right (i.e., $t$ begins with $1$ and ends with $n$) as follows. 

First note that the column $v'_t$ with $ t \in \{1,2,3,5,7,\cdots, 2^{m}-1\}$ requires no more treatment, since it has been in the form of  a column of $I$, that is, $v'_t=\Vec{e}_{t}$. 

Now given $v'_t$ with  $t \in \{4,6,8,\cdots, 2^m\}$, if $v'_{t,t}=1$, then no operation will be done. Otherwise, according to the property (2) of $V'$, it must be in the form: $v'_t=\Vec{e}_{t-1} \oplus \sum_{j=1}^{t-2}{k_j \Vec{e}_{j}}$ for at least one $k_j$ being 1. In other words, $v'_{t-1,t}=1$ and there exists $ t' < t-1$ such that $v'_{t',t}=1$. Then we  apply a Toffoli gate on the bit $q_t$, with  control bits being $q_{t-1}, q_{t'}$, which changes $v'_{t,t}$ to be $1$. Note that this gate will not change $v'_{l<t}$ (the columns before $v'_t$) by observing the following two points:  (i) by the property (1) of $V'$, $v'_{t-1}=\Vec{e}_{t-1}$ will not be affected by the Toffoli gate;  (ii) for $v'_1, v'_2, \cdots,  v'_{t-2}$, their $(t-1)$-th bits are all 0 due to the fact that $V'$ is an upper triangular matrix, and thus they will not be affected by the Toffoli gate too. Therefore, the matrix is still  upper triangular after the above Toffoli gate. 

 We now consider the effect of the above Toffoli gate on $v'_{l>t}$ (the columns after $v'_t$). If the Toffoli gate doesn't work on $v'_l$, the properties of $V'$  still hold. Obviously, the Toffoli gate has no effect on the odd column of $V'$, since they have only one $1$. 
 
 Suppose $v'_l$ is affected by the Toffoli gate generated by $v'_t$, then it must hold that $l\geq t+2$ and $l \in \{4,6,8,\cdots, 2^m\}$. 
 It's easy to check $v'_l=\Vec{e}_{l-1} \oplus \sum_{j=1}^{l-2}{k_j \Vec{e}_{j}}$ still holds for at least one $k_j$ being 1 after the Toffoli gate ($k_{t-1}, k_{t'}$ are 1 and $k_t$ flips). 
 Therefore, the properties of $v'_{l>t}$ are not changed by the Toffoli gate generated by $v'_t$. As a result, we can similarly deal with the next column $v'_{(t+1)}$.

 Because there are no more than $2^{m-1}-1$ vectors satisfying $v_t\neq u_t$, this procedure will cost no more than $2^{m-1}-1$ Toffoli gates, which means no more than $O(n)$ elementary gates are needed.  After the above procedure, $V'$ has been transformed to an upper triangular matrix $V''$ with the main diagonal elements being $1$, and thus is a full rank matrix.  Then we can use another $n\times n$ reversible matrix $R_2$ over $ \mathbb{F}_2$ such that $R_2 V''=I$. Note that $R_2$ consumes  $O(n^2/\log n)$ CNOT gates by Lemma \ref{thorem:res}.
 Therefore, the total gate cost is $O(n^2/\log n) +O(n) =O(n^2/\log n)$.
\end{proof}

The three Lemmas above can be summarized as follows.

\begin{lemma}\label{Pelemma}
    Let $m=\lfloor\log_2n\rfloor-1$. When $|S_P|>2^{n/2-1}$, $2^{m-1}$ pairs of $(i, x_i)$ from $S_P$ can be selected, for which we can use no more than $O(n^2/\log n)$ elementary gates to construct an reversible operator $P_e$  satisfying $P_e\ket{i_t}=\ket{e_{2t-1}}$ and $ P_e\ket{x_{i_t}}=\ket{e_{2t}}$ for any selected pair $(i_t, x_{i_t})$.
\end{lemma}

\begin{proof}
    The lemmas above have shown that $P_e$  costs two $n\times n$ reversible matrices and no more than $2^{m-1}-1$ Toffoli gates.  Therefore, it will cost no more than $O(n^2/\log n)$ elementary gates to implement $P_e$.
\end{proof}

After proving these lemmas, we obtain the following crucial result.

\begin{theorem}\label{mainswap}
    For a given $2^n \times 2^n$ permutation matrix P, there exists a procedure using no more than $O(2^n n/\log n )$ elementary gates to implement a permutation matrix $P_Q$, such that $|S_{P_Q P}|\leq 2^{n/2-1}$.
    
\end{theorem}
\begin{proof}
    The procedure is as follows:
    \begin{enumerate}[(1)]
        \item 
        Let $m=\lfloor\log_2n\rfloor-1$. Take $2^{m-1}$ pairs  $\{(i_t, x_{i_t}): t=1,2,\cdots,  2^{m-1} \}$ from $S_P$ which satisfy the requirement in Lemma \ref{lemma-matrix}. Let $\ket{a_t}=\ket{i_t}$, $ \ket{b_t} = \ket{x_{i_t}}(\forall t \in \{1,2,\cdots,  2^{m-1}\})$. If no such $2^{m-1}$ pairs can be selected, the procedure ends.
        
        \item
        Construct a reversible operator $P_e$, which satisfies $\forall t \in \{1,2,\cdots,  2^{m-1}\}$, $P_e\ket{a_t}=\ket{e_{2t-1}}$, $P_e\ket{b_t}=\ket{e_{2t}}$. 
        {\color{ black} Since $P_e$ is reversible, for those $x\notin \{a_1, b_1, \cdots, a_{2^{m-1}}, b_{2^{m-1}}\}$, $P_e\ket{x} \notin \{\ket{e_1}, \ket{e_2}, \cdots, \ket{e_{2^m}}\}$.}
        
        \item
        Construct a reversible operator $P_{map}$, which satisfies $P_{map}\ket{e_i}=\ket{y_i}(\forall i \in \{1,2,\cdots,  2^{m}\})$, where $\ket{y_i}=\ket{10\cdots 0}_{n-m}\ket{i-1}_m$ and $\ket{i-1}_m$ means the $m$-bit binary string of $i-1$. 
        {\color{ black} Since $P_{map}$ is reversible, for those $x\notin \{e_1, e_2, \cdots, e_{2^m}\}$, $P_{map}\ket{x} \notin \{\ket{y_1}, \ket{y_2}, \cdots, \ket{y_{2^m}}\}$.}

        \item
        Deploy a C$^{n-m}$NOT gate on the bit $q_{n}$, with the control bits being the first $n-m$ bits $\ket{10\cdots 0}_{n-m}$. We call this step a main swap, and use $T_{ms}$ to represent its matrix form.

        \item
        Apply $P^{-1}_{map}$, $P^{-1}_e$.  After these operations mentioned before, we get an operator $Q=P^{-1}_e P^{-1}_{map} T_{ms} P_{map}P_e$, which satisfies
       \begin{equation}
       \begin{aligned}
            Q\ket{x}=\left\{
            \begin{array}{rcl}
                \ket{b_t}, & & x = a_t\\
                \ket{a_t}, & & x = b_t\\
                \ket{x}, & & else
            \end{array} \right. \ 
            \\  (\forall t \in \{1,2,\cdots,  2^{m-1}\}) \ .
        \end{aligned}
        \end{equation}
        
        \item
         Update $P=QP$. Return to Step 1.
         
    \end{enumerate}
    {\color{ black}We now explain why $Q\ket{x}=\ket{x}$ for other values of $x$.
    Combining step (2) and step (3), we have $P_{map}P_e\ket{x}\notin \{\ket{y_1}, \ket{y_2}, \cdots, \ket{y_{2^m}}\}$, and thus $Q \ket{x}=P^{-1}_e P^{-1}_{map} P_{map}P_e\ket{x}=\ket{x}$.}
    \begin{figure}[h]
        \centering
        \includegraphics[width=1\textwidth]{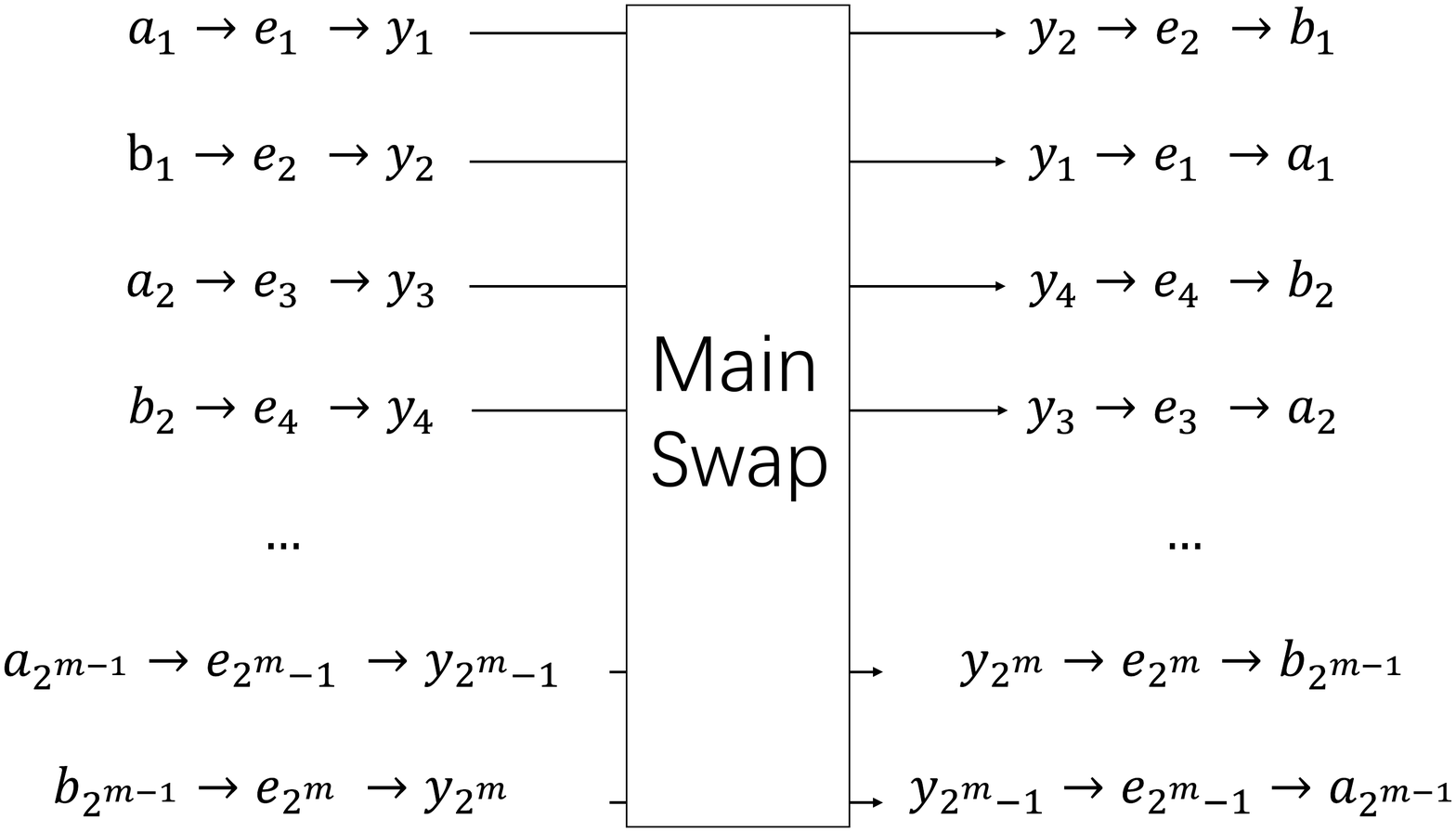}
        \caption{Framework of Theorem \ref{mainswap}}
        \label{p1}
    \end{figure} 
    The framework of this procedure is shown in Figure \ref{p1}. After each iteration, we have $QP\ket{a_t}=\ket{a_t}$ for all $ t \in \{1,2,\cdots,  2^{m-1}\}$. In addition, for $ i$ satisfying $P\ket{i}=\ket{i} $, we have $QP\ket{i}=\ket{i}$. Thus, $|S_P| - |S_{QP}| \geq 2 ^{m-1} \geq n/8$, which means the procedure won't work more than $8\cdot 2^n/n$ iterations. 

    According to Lemma \ref{Pelemma}, $P_e$ needs $O(n^2/\log n)$ elementary gates. 
    
  Here we  show how to construct $P_{map}$. For each $e_i$, we use $q_i$ as the control bit to apply CNOT gate to flip $q_1$ such that $q_1=1$ and the last $m$ bits to be the binary string of $i-1$. The result is that we transform $\ket{e_i}$ to $\ket{10\cdots0}_{i-1}\ket{10\cdots0}_{n-m-i+1}\ket{i-1}_m$.  The cost of this transformation is no more than $m+1$ gates. Then apply a C$^{m}$NOT gate on bit $q_i$ with control bits being the last $m$ bits, which flips the $i$-th bit when the last $m$ bits are $\ket{i-1}_m$.  As a result,  $\ket{e_i}$ is transformed to $\ket{y_i}$. Because $2^m<n-m$, these gates have no effect on $\ket{y_j}$ and $\ket{e_{j \neq i}}$ for any $j$.  Since Lemma \ref{costofCNOT} shows these C$^m$NOT gates have a linear cost, $P_{map}$ costs no more than $O(2^m m)=O(n\log n)$ elementary gates.
      {\color{ black} Without losing generality, we can assume that $m>1$. According to Lemma \ref{costofCNOT}, the cost of C$^{n-m}$NOT gate in step (4) is $O(n-m)$.}
    
    Therefore, each $Q$ needs no more than $O(n^2/\log n+n\log n)=O(n^2/\log n)$ elementary gates. Let $P_Q=Q_d\cdots Q_2 Q_1$, where $d$ is the number of iterations and has been  shown to be less than $ 8\cdot 2^n/n$.  Therefore, the total cost of this procedure is no more than $O(2^n n/\log n)$. 

    When this procedure ends, it means we can't select $2^{m-1}$ pairs as Lemma \ref{lemma-matrix} requires, then we get $|S_{P_Q P}|\leq 2^{n/2-1}$.
\end{proof}

\noindent \textbf{Example.} The following example illustrates our method. Suppose $n=8$, and thus $m=2$. The reversible function we want to implement is as follows:
    \begin{equation}
             f(x)=\left\{
            \begin{array}{rcl}
                0, & & x=0\\
                1, & & x=254\\
                2, & & x=255\\
                x+2,& & else
            \end{array} \right. 
    \end{equation}

(1) We select the pair $(1,3)$, then we have
\begin{equation}
    U=
\begin{bmatrix}
    0& 0\\
    0& 0\\
    0& 0\\
    0& 0\\
    0& 0\\
    0& 0\\
    0& 1\\
    1& 1\\
\end{bmatrix},
V=
\begin{bmatrix}
    0& 0\\
    0& 0\\
    0& 0\\
    0& 0\\
    0& 0\\
    0& 0\\
    0& 1\\
    1& 1\\
\end{bmatrix}.
\end{equation}

(2) Since $\Vec{1}$, $\Vec{2}$, $\Vec{3}$ can be represented by the linear combination of $U$, the pairs $(254, 1)$, $(255, 2)$, $(1, 3)$, $(2, 4)$, $(3, 5)$ are invalid. Now we select $(4, 6)$, and update
\begin{equation}
    V=
    \begin{bmatrix}
        0& 0& 0& 0\\
        0& 0& 0& 0\\
        0& 0& 0& 0\\
        0& 0& 0& 0\\
        0& 0& 0& 0\\
        0& 0& 1& 1\\
        0& 1& 0& 1\\
        1& 1& 0& 0\\
    \end{bmatrix}.
\end{equation}
 Since $\Vec{6}$ is the linear combination of $\{\Vec{1}, \Vec{3}, \Vec{4}\}$, we take $\Vec{8}$ as $u_4$ and update
\begin{equation}
    U=
    \begin{bmatrix}
        0& 0& 0& 0\\
        0& 0& 0& 0\\
        0& 0& 0& 0\\
        0& 0& 0& 0\\
        0& 0& 0& 1\\
        0& 0& 1& 0\\
        0& 1& 0& 0\\
        1& 1& 0& 0\\
    \end{bmatrix}.
\end{equation}

(3) After $2^{m-1} = 2$ pairs have been selected, augment  $U$ so that it is a full rank matrix 
\begin{equation}
    U=
    \begin{bmatrix}
        0& 0& 0& 0& 0& 0& 0& 1\\
        0& 0& 0& 0& 0& 0& 1& 0\\
        0& 0& 0& 0& 0& 1& 0& 0\\
        0& 0& 0& 0& 1& 0& 0& 0\\
        0& 0& 0& 1& 0& 0& 0& 0\\
        0& 0& 1& 0& 0& 0& 0& 0\\
        0& 1& 0& 0& 0& 0& 0& 0\\
        1& 1& 0& 0& 0& 0& 0& 0\\
    \end{bmatrix}.
\end{equation}
Then, augment  $V$ with the same vectors. We have 
\begin{equation}
    V=
    \begin{bmatrix}
        0& 0& 0& 0& 0& 0& 0& 1\\
        0& 0& 0& 0& 0& 0& 1& 0\\
        0& 0& 0& 0& 0& 1& 0& 0\\
        0& 0& 0& 0& 1& 0& 0& 0\\
        0& 0& 0& 0& 0& 0& 0& 0\\
        0& 0& 1& 1& 0& 0& 0& 0\\
        0& 1& 0& 1& 0& 0& 0& 0\\
        1& 1& 0& 0& 0& 0& 0& 0\\
    \end{bmatrix}.
\end{equation}

(4) Construct   an  $n\times n$ reversible matrix $R_1$ which satisfies $R_1 U = I$. Meanwhile, 
\begin{equation}
    V'=R_1 V=
    \begin{bmatrix}
        1& 0& 0& 1& 0& 0& 0& 0\\
        0& 1& 0& 1& 0& 0& 0& 0\\
        0& 0& 1& 1& 0& 0& 0& 0\\
        0& 0& 0& 0& 0& 0& 0& 0\\
        0& 0& 0& 0& 1& 0& 0& 0\\
        0& 0& 0& 0& 0& 1& 0& 0\\
        0& 0& 0& 0& 0& 0& 1& 0\\
        0& 0& 0& 0& 0& 0& 0& 1\\
    \end{bmatrix}.
\end{equation}

(5) Apply a Toffoli gate on $q_4$, with the control bits being $q_1$ and $q_3$. Then we have 
\begin{equation}
    V''=
    \begin{bmatrix}
        1& 0& 0& 1& 0& 0& 0& 0\\
        0& 1& 0& 1& 0& 0& 0& 0\\
        0& 0& 1& 1& 0& 0& 0& 0\\
        0& 0& 0& 1& 0& 0& 0& 0\\
        0& 0& 0& 0& 1& 0& 0& 0\\
        0& 0& 0& 0& 0& 1& 0& 0\\
        0& 0& 0& 0& 0& 0& 1& 0\\
        0& 0& 0& 0& 0& 0& 0& 1\\
    \end{bmatrix}.
\end{equation} Since $V''$ is a full rank matrix now, we can construct another    $n \times n$ reversible matrix $R_2$ such that $R_2 V''=I$. Step 4 and Step 5 construct $P_e$.

(6) For $\ket{e_1}=\ket{y_1}=\ket{10000000}$, we don't need to do anything.

For $\ket{e_2}=\ket{01000000}$, apply CNOT gate on $q_1, q_8$ with the control bit being $q_2$ and we have $\ket{11000001}$.  Apply C$^2$NOT gate on $q_2$ with  the control bits being $\Bar{q_7}, q_8$, getting $\ket{y_2}=\ket{10000001}$.

For $\ket{e_3}=\ket{00100000}$, apply CNOT gate on $q_1, q_7$ with control bit is $q_3$ and we have $\ket{10100010}$.  Apply C$^2$NOT gate on $q_3$ with control bits being $q_7, \Bar{q_8} $, getting $\ket{y_3}=\ket{10000010}$.

For $\ket{e_4}=\ket{00010000}$, apply CNOT gate on $q_1, q_7, q_8$ with the control bit being $q_4$ and we have $\ket{10010011}$.  Apply C$^2$NOT gate on $q_4$ with the control bits being $q_7, q_8$, getting  $\ket{y_4}=\ket{10000011}$.

Obviously, the gates above have no effect on other $\ket{e_i}$ and $\ket{y_i}$.

(7) Apply a C$^6$NOT gate on $q_8$, with the control bits being the first 6 bits $\ket{100000}$. Then apply the gates above in a reverse order, and get an operator $Q$, which satisfies
\begin{equation}
    Q\ket{x}=\left\{
    \begin{array}{rcl}
        \ket{3} & & x=1\\
        \ket{1} & & x=3\\
        \ket{6} & & x=4\\
        \ket{4} & & x=6\\
        \ket{x} & & else
    \end{array} \right. .   
\end{equation}
Therefore,
\begin{equation}
    QP\ket{x}=\left\{
    \begin{array}{rcl}
        \ket{x} & & x=0,1,4\\
        \ket{3} & & x=254\\
        \ket{6} & & x=2\\
        \ket{2} & & x=255\\
        \ket{x + 2} & & else
    \end{array} \right. .
\end{equation}

Here we have $|S_{QP}|= |S_P|-2$. This is  how an iteration works. Then we update $P=QP$, and repeat the above procedure until we can't take $2^{m-1}$ pairs which satisfy the requirement of Lemma \ref{lemma-matrix}.

Now we give the proof of Theorem \ref{maintheorem}.

\noindent\textbf{Proof of Theorem \ref{maintheorem}.}
    For a given $n$-bit circuit, we use a $2^n \times 2^n$ permutation matrix $P$ to represent it. Theorem \ref{mainswap} proves that we can use no more than $O(2^n n/\log n)$ elementary gates to implement $P_Q$, such that $|S_{P_Q P}|\leq 2^{n/2-1}$.

    Let 
    \begin{equation}
        P_r=P_QP \ .
    \end{equation} We call $P_r$ the rest part, and have $|S_{P_r}|\leq 2^{n/2-1}$. Then we can use Lemma \ref{lemma2.2} to address the rest part, and the cost is no more than $O(2^{n/2}n^2)$. Hence, the total cost is $O(2^n n/\log n+2^{n/2}n^2 )=O(2^n n/\log n )$.

\section{Improving The  Decomposition of Rest Part}\label{sec4}
For the rest part $P_r$ of Theorem \ref{mainswap}, the algorithm provided by Lemma \ref{lemma2.2} is really trivial. In this section, we propose an algorithm using no more than $2^{n/2-1}(4n+64 \log n + O(1))$ elementary gates to implement it. The best algorithm so far requires no more than $2^{n/2-1}(8.5n+O(1))$ \cite{saeedi2010reversible}. Our algorithm will have a better performance when $n > 100$.

Since we want to compare with the best result so far, we need to use the same gate library. The gate library includes NOT, CNOT, Hadamard, controlled-V gate,with $V=\frac{i+1}{2} \begin{bmatrix}
    1 & -i\\
    -i & 1
\end{bmatrix}$ and rotation gate $R_a(\theta)$, which $a \in \{x,y,z\}$ and $\theta \in [0,2\pi]$. By Ref.\cite{maslov2008quantum}, for $n\geq 5$ and $m \in \{3,4,\cdots,\lfloor n/2 \rfloor \}$, a C$^m$NOT gate can be implemented using no more than $12m-22$ elementary gates. For $n \geq 7$ and $m \in \{\lfloor n/2 \rfloor +1,\lfloor n/2 \rfloor +2, \cdots ,n-2\}$, a C$^m$NOT gate can be implemented using no more than $24m-40$ elementary gates 

\begin{lemma}\label{lemma_distinct}
    When $|S_P|>3 \cdot 2^{m-1}-3$, $2^{m-1}$ pairs of $(i, x_i)$ can be selected, such that all of them are distinct, where $m$ is an  integer satisfying $0\leq m\leq\lfloor\log_2n\rfloor$.
\end{lemma}

\begin{proof}
    Suppose we are selecting the $t$-th pair, which means $t-1$ pairs have been selected. Each selected pair $(i, x_i)$ would make $(i,x_i),(x_i,f(x_i)),(f^{-1}(i),i)$ unable to be selected. Hence there are at most $3t-3$ pairs  unavailable at this moment. Thus, when $|S_P|>3 \cdot 2^{m-1}-3$,  $2^{m-1}$ pairs of $(i, x_i)$ could be selected.
\end{proof}

\begin{theorem}
    For the rest part $P_r$ of Theorem \ref{mainswap} with $|S_{P_r}|\leq 2^{n/2-1}$, there exists an algorithm requiring no more than $2^{n/2-1}(4n + 64 \log n + O(1))$ elementary gates to implement it.
\end{theorem}

\begin{proof}

    The framework of this algorithm is the same as the one in Theorem \ref{mainswap}, but needs to slightly modify $P_e$ and $P_{map}$. In this algorithm, our $2^{m-1}$ pairs of $(i, x_i)$ only need to satisfy all the numbers are distinct. We set $m=\lfloor\log_2n\rfloor$ at the beginning.

    For $P_e$, it requires no more than $2^m(n+O(1))$ elementary gates. We need to transform the selected ones to $e_t$ one by one. We will do it in the ascending order. 
    
    When we are turning $v$ to $e_t$, check whether the $t$-th bit is 1. If not, since $v \neq 0$, we can find at least one $t_1$ satisfying the $t_1$-th bit is 1. If we can't find another $t_2\neq t_1$ satisfying the $t_2$-th bit is 1, 
    {\color{ black}we have $t_1>t$ since $v$ can't be $e_{t'<t}$}.
    Apply a CNOT gate on the $t$-th bit with $q_{t_1}$ being the control bit. Otherwise, use a Toffoli gate on the $t$-th bit with $q_{t_1}, q_{t_2}$ being the control bits. In other words, the maximum cost is $2^m - 1$ Toffoli gates, which come at a elementary gate cost of $2^m \cdot O(1)$. 
    {\color{ black}The $O (1) $ here is the cost of Toffoli gates. Depending on their construction methods, the elementary gate cost may vary. For example, \cite{yu2013five} describes a method with a minimum two-bit gate overhead, while \cite{amy2013meet} describes a method with a minimal depth. Here we will use $O(1)$ to represent its cost and treat it as a constant-cost quantum gate.}

    When the $t$-th bit is 1, apply no more than $n-1$ CNOT gates to turn the other bits to 0, whose control bit is $q_t$. Obviously, these gates have no effect on $e_{t'<t}$. 
    {\color{ black}Thus, $P_e$ can be constructed with no more than $2^m(n-1) + 2^m \cdot O(1) = 2^m(n +O(1))$ elementary gates.}

    For $P_{map}$, it's possible that $2^m > n-m$, so we need to give another way to transform $\ket{e_{t>n-m}}$ into $ \ket{ y_{t>n-m}}$. Let $\ket{e'_t}=\ket{10\cdots0}_{i}\ket{10..0}_{n-t-i}\ket{10..0}_{t}$, in which $t=n-m+i$ and $i \in \{1, 2,\cdots, 2^m-(n-m)\}$. We apply 2 CNOT gates on $q_1$ and $q_{i+1}$, whose control bit is $q_{t}$. After transforming all $e_{t'>n-m}$ to $e'_{t'>n-m}$, apply no more than $m$ Toffoli gates to transform the last $m$ bits to be $\ket{t'-1}_m$ with control bits being $q_1, q_{i+1}$.  Then apply C$^m$NOT gates on $q_{i+1}$ with the last $m$ bits being the control bits, which only flips the target bit when the last $m$ bits are $\ket{t'-1}_m$. Since these $m$ Toffoli gates have common control bits, {\color{ black}they can be done by $2m-2$ CNOT gates and one Toffoli gate \cite{maslov2005toffoli}, which requires $2m +O(1)$ elementary gates.
    Similarly, the $O(1)$ here is the cost of a Toffoli gate.}
    After that,  just follow Theorem \ref{mainswap} to deal with those $e_{j\leq n-m}$.

    After that, all the $e_{t'>n-m}$ have been transformed to $y_{t'>n-m}$. For each $e_{t'>n-m}$, it requires no more than $m$ Toffoli gates, a C$^m$NOT gate and two CNOT gates.  For those $e_{j\leq n-m}$, each needs no more than $m+1$ CNOT gates and one C$^m$NOT gate.
    
    Because $m =\lfloor \log_2n \rfloor \leq \lfloor n/2 \rfloor$, C$^m$NOT gate needs no more than $14m-22$ elementary gates. Thus, the cost of $P_{map}$ is less than $2^m(16m+O(1))$.

    {\color{ black} Similar to step (4) in Theorem \ref{mainswap}, we need to use a C$^{n-m}$NOT gate to implement the main swap.
    Since $n-m > n/2$, the cost of a general C$^{n-m}$NOT gate is $26(n-m)+O(1)$.}

    Therefore, the cost of an iteration is no more than $2^m(2n + 32m + O(1)) + 26(n-m) + O(1)$. 
    Each iteration removes at least $2^{m-1}$ elements from $S_{P_r}$. 
    {\color{ black} Therefore, removing each element from $S_{P_r}$ takes no more than $4n + 64m +26 (n-m)/2 ^ {m-1} +O(1) = 4n + 64m + O(1)$ elementary gates  when $|S_{P_r}|$ is large enough (larger than $1.5n$).}
    
    Lemma \ref{lemma_distinct} shows that  when the required $2^{m-1}$ pairs are unavailable for the first time, we have $|S_{P_r}|<1.5n$, which won't affect the gate complexity. 
    
    We could reduce $m$ by 1, and repeat this algorithm until $m=0$. 
    {\color{ black} When $0<m<\lfloor \log_2n \rfloor$, the total gate cost is no more than 
    \begin{align}
        &\sum_{m=1}^{\lfloor \log_2n \rfloor -1} 1.5(2^{m+1} - 2^m)\cdot (4n + 64m +26 (n-m)/2 ^ {m-1} +O(1))\\
        &\leq \sum_{m=1}^{\lfloor \log_2n \rfloor -1}2^m (6n + 96m + O(1))+78(n-m) \\
        &\leq n\log_2n(6n + 96\log_2n + O(1)) +78n\log_2n
    \end{align}
    Obviously, this term will not affect the overall cost.
    }    
    When $m=0$, $|S_{P_r}|\leq 3$, this algorithm uses no more than 2 C$^{n-1}$NOT gates. Therefore, the total gate number is no more than 
    $2^{n/2-1}(4n + 64m + O(1))$.
\end{proof}

Here we give an example to help readers understand $P_{map}$ better. Suppose $n$=8 and $m$ =3.  We have to deal with $e_6, e_7, e_8$ first. We apply CNOT gates on $q_1, q_2$ with control bit $q_6$. Then we have $\ket{e'_6}=\ket{11000100}$. Similarly, $\ket{e'_7}=\ket{10100010}$ and $\ket{e'_8}=\ket{10010001}$. Then we use $q_1,q_2$ as control bits to apply a Toffoli gate on $q_8$ to get $\ket{11000101}$. After that, apply a C$^3$NOT gate on $q_2$ with the control bits being $q_6,\Bar{q_7},q_8$ to get $\ket{y_6}=\ket{10000101}$. Obviously, the above gates have no effect on $e'_{t\neq 6}, e_{t\neq 6}, y_{t\neq 6}$. By this pattern, we have $y_6, y_7. y_8$ and leave $e_1, e_2, e_3, e_4 ,e_5$ unchanged.

For $e_1, e_2, e_3, e_4 ,e_5$, just follow the steps in Theorem \ref{mainswap}, and those gates won't work on $y_6,y_7,y_8$.

\section{Conclusion}\label{secCon}
In this paper, we have  proved that any arbitrary $n$-bit reversible circuit can be implemented with no more than $O(2^n n/\log n)$ elementary gates without using any ancillary bits, thus achieving the lower bound $\Omega(2^n n/\log n)$  proposed in 2003 and closing the open problem in \cite{saeedi2013synthesis}. In the future, one may  consider more efficient methods, since  our method  costs much time to check  linear independence.

\section*{Acknowledgments}
The authors would like to thank  the anonymous re viewers for their thoughtful comments which helped greatly improve the presentation. This work was supported by the National Natural Science Foundation of China (Grant No. 62272492), the Guangdong Provincial Quantum Science Strategic
Initiative (Grant  No. GDZX2303007), and the Guangdong Basic and Applied Basic Research Foundation (Grant No. 2020B1515020050).

\bibliographystyle{unsrt}
\bibliography{sample}

\end{document}